\documentclass[english]{elsarticle}
\usepackage[T1]{fontenc}
\usepackage[latin9]{inputenc}
\usepackage{amsthm}
\usepackage{amsmath}
\usepackage{amssymb}
\usepackage{graphicx}

\makeatletter

\providecommand{\tabularnewline}{\\}

\theoremstyle{plain}
\newtheorem{thm}{\protect\theoremname}
\theoremstyle{definition}
\newtheorem{defn}[thm]{\protect\definitionname}
\theoremstyle{plain}
\newtheorem{prop}[thm]{\protect\propositionname}
\ifx\proof\undefined
\newenvironment{proof}[1][\protect\proofname]{\par
\normalfont\topsep6\p@\@plus6\p@\relax
\trivlist
\itemindent\parindent
\item[\hskip\labelsep\scshape #1]\ignorespaces
}{%
\endtrivlist\@endpefalse
}
\providecommand{\proofname}{Proof}
\fi

\journal{Journal of Mathematical Psychology}

\makeatother

\usepackage{babel}
\providecommand{\definitionname}{Definition}
\providecommand{\propositionname}{Proposition}
\providecommand{\theoremname}{Theorem}

\begin{document}

\begin{frontmatter}{}

\title{Negative Probabilities and Contextuality}

\author{J. Acacio de Barros}

\address{School of Humanities and Liberal Studies, San Francisco State University,
San Francisco, CA, USA}

\author{Janne V. Kujala}

\address{Department of Mathematical Information Technology, University of
Jyväskylä, Jyväskylä, Finland}

\author{Gary Oas}

\address{Stanford Pre-Collegiate Studies, Stanford University, Stanford, CA,
USA}
\begin{abstract}
There has been a growing interest, both in physics and psychology,
in understanding contextuality in experimentally observed quantities.
Different approaches have been proposed to deal with contextual systems,
and a promising one is contextuality-by-default, put forth by Dzhafarov
and Kujala. The goal of this paper is to present a tutorial on a different
approach: negative probabilities. We do so by presenting the overall
theory of negative probabilities in a way that is consistent with
contextuality-by-default and by examining with this theory some simple
examples where contextuality appears, both in physics and psychology. \end{abstract}

\end{frontmatter}{}

\section{Introduction\label{sec:Introduction}}

In recent years there has been an increased interest in modeling psychological
experiments with the mathematical tools of Quantum Mechanics (QM)
\cite{busemeyer_quantum_2012}. The argument, already put forth by
Bohr, is that the principle of complementarity in QM is not unique
to physical events, but is also present in cognitive and social phenomena
\cite{holton_roots_1970}. Since complementarity is ubiquitous, it
should also be true that the Hilbert space formalism created by physicists
at the beginning of the 20th Century can be applied to describe mathematical
situations outside of physics. This line of thinking gave rise in
recent times to a thriving line of research known as \emph{Quantum
Interaction} and, more specifically in the context of psychology,
\emph{Quantum Cognition. }

At the core of complementarity is the idea that it is not possible,
in principle, to observe simultaneously certain characteristics of
a system. In physics, this is the case for the well-known wave/particle
duality: each experimental context determines which characteristic,
undulatory or corpuscular, is observed. Similarly, it was proposed
that complementarity in psychology appears when a subject has to deal
with situations that have different and incompatible contexts. The
key aspect of complementarity, for our purpose, is that of a dependency
on the context. Therefore, quantum mathematical models, and also quantum
cognition, are essentially descriptions of contextual observables. 

That the mathematical apparatus of QM is well suited to describe those
context-dependent observables found in physics is clear by the tremendous
success of this theory to not only describe the microscopic world
but also to predict surprising results. This success comes from the
fact that complementarity, with its prohibition of simultaneous observation
of certain quantities, implies an orthomodular lattice of propositions
pertaining to the observable events, instead of a classic Boolean
algebra of compatible observables. In a famous paper, Piron \cite{piron_axiomatique_1964}
proved that the orthomodular lattice of propositions have a representation
in terms of Hilbert spaces. Therefore, it stands to reason that Hilbert
spaces are a good candidate for modeling the probabilities of quantities
that may not be simultaneously observable. In other words, the mathematics
of QM is a well-suited extension of probability theory that offer
a way to model the probabilistic outcomes of contextual observables
\cite{blutner_quantum_2015}. 

However, the mathematical structure of QM does not come without a
price. First, it is not the most universal generalization of probabilities
for context-dependent systems. It is possible to imagine certain context-dependent
situations of interest to researchers outside of physics which the
Hilbert space formalism of QM fails to describe (see \cite{de_barros_beyond_2015}
for an example). Second, the quantum formalism predicts some results
that are not reasonable in, say, psychology. For instance, one important
result is the impossibility to clone an unknown quantum system, which
is related to the impossibility of superluminal signaling. There is
no analogue to this in psychology, and one should not expect the cloning
of ``cognitive states'' to be impossible in principle (for instance,
in principle, albeit not in practice, we could conceive of duplicating
all the neural states of a given brain, with their corresponding firings
and configurations). 

It is thus reasonable to ask what other ways of describing contextual
systems exist. This has been a matter of intense research in the past
few years, and in this paper we provide one possible tool: Negative
Probabilities (NP). Our purpose here is to lay out the main ideas
necessary to describe certain contextual systems with NP. To do so,
we organize this paper as follows. In Section \ref{sec:Contextuality},
we start with a definition of contextuality, in line with the recent
work of Dzhafarov and Kujala \cite{dzhafarov_contextuality_2014-1}.
In Section \ref{sec:Negative-Probabilities} we go into the mathematical
details of negative probabilities, and discuss possible interpretations.
Finally, in Section \ref{sec:Some-examples} we present some examples
and applications of NP.

\section{Contextuality\label{sec:Contextuality}}

To understand what we mean by contextuality, we need to lay down some
notation to describe it. Let us start with a formal definition of
probabilities, which will be useful later on, when we modify it to
allow for contextual systems. We follow Kolmogorov's axiomatic approach
based on set theory in general\cite{kolmogorov_foundations_1956},
but for the present paper, will only need to use finite probability
spaces.

A discrete probability space is determined by the triple $(\Omega,\mathcal{F},p)$,
where $\Omega$ is the set of elementary events, $\mathcal{F}$ is
the algebra of events (which can be taken as the powerset $2^{\omega}$
for our purposes), and $p:\mathcal{F}\to[0,1]$ is a function that
yields the probability of each event $S\in\mathcal{F}.$ The elementary
events define the most atomic outcomes of an experiment, and so the
probability of a general event $S$ is determined by the probabilities
of elementary events: $p(S)=\sum_{\omega\in S}p(\{\omega\}).$ 

It is important to note that one generally cannot observe the elementary
events directly. Let us explain what we mean with the well-known firefly
box \cite{foulis_half-century_1999}, which will be useful later on
when we introduce the concept of contextuality. Imagine we have a
box with a firefly inside it emitting light at random times. The box
is constructed such that its walls are translucent, but an observer
can only see one side of the box at a time (see Figure \ref{fig:Foulis's-box}).
\begin{figure}
\begin{centering}
\includegraphics[scale=0.15]{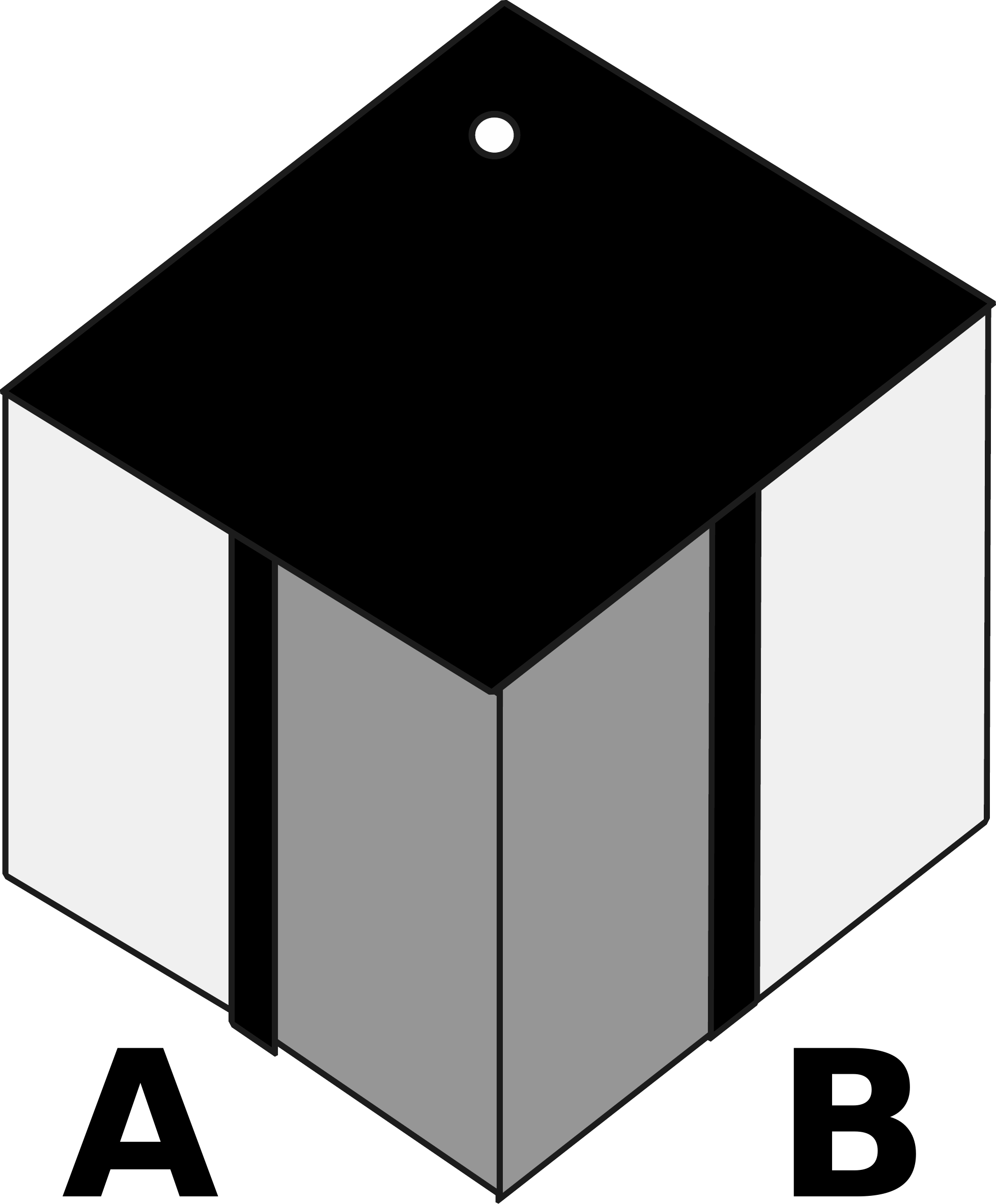}
\par\end{centering}

\protect\caption{\label{fig:Foulis's-box}A firefly inside the box, whose position
is represented on the horizontal plane by the white dot on top of
the box, shines its light at a certain instant of time. The box is
designed such that if the firefly is on the left hand side of A, then
only this side lights up, but the exact position of the insect on
this plane cannot be inferred (similarly for B). Due to experimental
constraints (also by design, we cannot look at both sides A and B
simultaneously), even though the actual position of the firefly is
given by, in the figure, left on A and right on B, we only know one
at a time, but do not know them jointly (i.e., knowing A is left does
not tell us what the value of B is, which in this case is either right
or left). }
\end{figure}
 A \emph{possible} $\Omega$ may be the set of all possible joint
values of A and B, namely $\left\{ RR,RL,LR,LL\right\} $, where $RR$
corresponds to the firefly lighting up the right side of A and right
side of B, $RL$ to right side of A and left of B, and so on\footnote{Here, for our purposes, we discard the possible state where the firefly
is not blinking.}. For this $\Omega$, the elementary event $RR$ is never actually
observed, since to observe it means seeing both sides of the box simultaneously,
which is forbidden by experimental design. 

In the firefly box, the set of elementary events could be even more
fine grained. For instance, it could be $\left\{ T_{R}T_{R},T_{R}T_{L},\ldots,B_{L}B_{R},B_{L}B_{L}\right\} $,
where, e.g., $T_{R}B_{L}$ corresponds to the firefly being on the
top of the right side of A and on the bottom of the left side on B.
Then, if one observed the firefly on the left of A, any of the following
elementary events might be result in this: $T_{L}T_{R}$, $T_{L}T_{L}$,$T_{L}B_{R}$,$T_{L}B_{L}$,
$B_{L}T_{R}$, $B_{L}T_{L}$,$B_{L}B_{R}$,$B_{L}B_{L}$. If we were
to compute the probability of the event ``left of A'' happens, we
would have to take the conjunction of all those elementary events,
which once again are not directly observable. 

The previous discussion motivates the idea of a random variable, a
very important tool in modeling experimental outcomes. Intuitively,
random variables (r.v.) are mathematical representations of outcomes
of an experiment which may be stochastic, such as the outcomes of
the firefly box, which are only ``left on A,'' ``right on A,''
``left on B,'' or ``right on B,'' abbreviated by $L_{A}$, $R_{A}$,
$L_{B}$, and $R_{B}$, respectively. Random variables model this
experiment in the following way. We start with a probability space,
whose elementary events in $\Omega$ are sampled according to $p$.
For this probability space, we choose functions $\mathbf{A}:\Omega\rightarrow\left\{ -1,1\right\} $
and $\mathbf{B}:\Omega\rightarrow\left\{ -1,1\right\} $ such that
for a random sampling of elementary events $\omega\in\Omega$ following
$p$, the probabilities of the outcomes $\mathbf{A}(\omega)=-1$,
$\mathbf{A}(\omega)=1$, $\mathbf{B}(\omega)=-1$, and $\mathbf{B}(\omega)=1$,
(which are given by respectively $p(\mathbf{A}=-1)$, $p(\mathbf{A}=1)$,
$p(\mathbf{B}=-1)$, and $p(\mathbf{B}=1)$), are the same as the
probabilities of observing ``left on $A$'', ``right on $A$'',
``left on $B$'', ``right on $B$'', respectively. In other words,
what random variables do is set a partition on $\Omega$ such that
each element of this partition (which is in $\mathcal{F}$) corresponds
to an outcome of the experiment with the same probabilistic features.
Thus, a discrete random variable is formally a function $\Omega\to E$
from the probability space to a certain set $E$ of possible values.
For our example above, the random variables $\mathbf{A}$ and $\mathbf{B}$
are $\pm1$-valued, with $E=\left\{ -1,1\right\} $. 

The \emph{expected value of a random variable} $\mathbf{R}$ or, for
short, the \emph{expectation of} $\mathbf{R}$, on a probability space
$\left(\Omega,\mathcal{F},p\right)$, denoted $E\left(\mathbf{R}\right)$,
is defined as
\[
E\left(\mathbf{R}\right)=\sum_{\omega\in\Omega}p\left(\{\omega\}\right)\mathbf{R}\left(\omega\right).
\]
For two random variables $\mathbf{R}$ and $\mathbf{S}$, their \emph{moment}
is defined as the expectation of their product, $E\left(\mathbf{RS}\right)$,
and for three random variables $\mathbf{R}$, $\mathbf{S}$, and $\mathbf{T}$,
their \emph{triple moment} is defined as the expectation of the triple
product, $E\left(\mathbf{RST}\right)$. Higher moments, are defined
in the same way, as product expectations of four or more random variables.

So, the question is whether we can create random variables $\mathbf{A}$
and $\mathbf{B}$ that model the firefly box. What we mean here is
whether there exists a probability space $\left(\Omega,\mathcal{F},p\right)$
and discrete random variables on this space such that all statistical
characteristics of the outcomes of observations of the box are the
same as the statistical characteristics of the random variables. For
example, if we observe ``left on A'' 50\% of the time, then it must
be the case that $E\left(\mathbf{A}\right)=0$ for a r.v. $\mathbf{A}$
taking values $\pm1$. Notice however that because we only observe
$\mathbf{A}$ or $\mathbf{B}$, we cannot know what the value of the
second moment $E\left(\mathbf{AB}\right)$ is, and any probability
spaces and r.v.'s on them satisfying the observed marginals $E\left(\mathbf{A}\right)$
and $E\left(\mathbf{B}\right)$ would be adequate. It is easy to prove
that for this firefly box, we can always find a $\left(\Omega,\mathcal{F},p\right)$
consistent with all observed marginals\footnote{For instance, we can just choose a $\left(\Omega,\mathcal{F},p\right)$
such that $\mathbf{A}$ and $\mathbf{B}$ are statistically independent,
i.e. $E\left(\mathbf{AB}\right)=0$, since the moment is not observable
by construction. }.

However, a common probability space does not always exist for r.v's
representing a collection of properties that cannot all be observed
simultaneously. To see this, let us consider a slightly more complicated
firefly example. Imagine a box, shown in Figure \ref{fig:Three-sided-Foulis-box},
\begin{figure}
\begin{centering}
\includegraphics[scale=0.4]{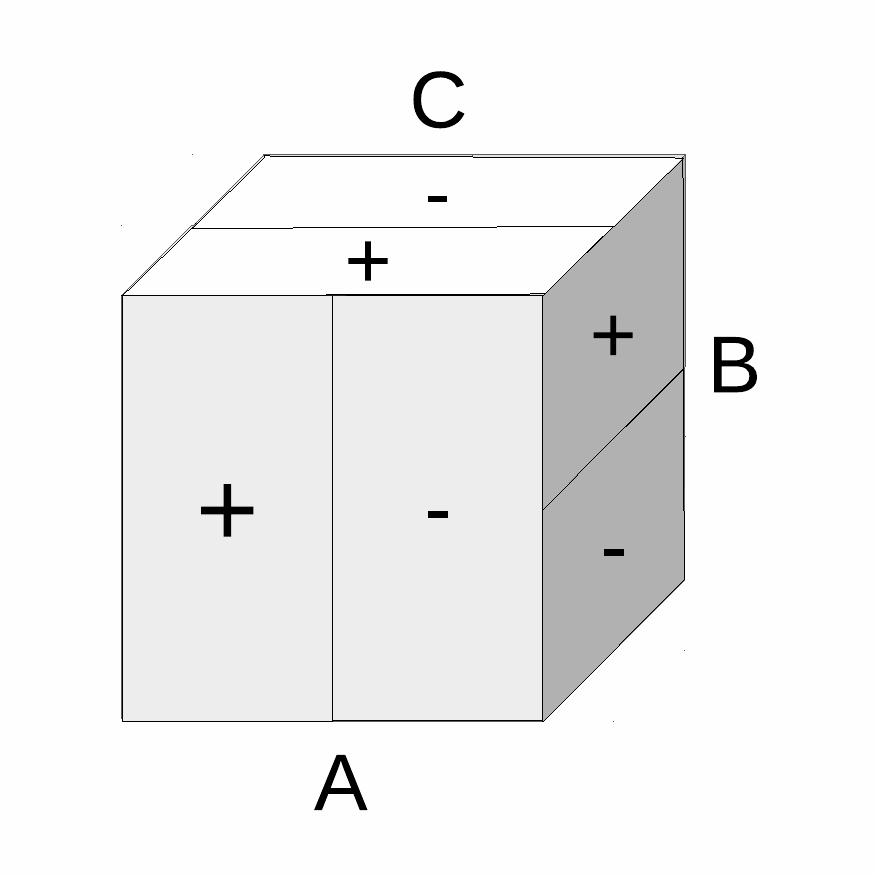}
\par\end{centering}

\protect\caption{\label{fig:Three-sided-Foulis-box}Three-sided firefly box. As before,
a firefly is present inside the box, blinking at random times. When
it blinks, because of the translucency of the walls, only one of face's
side lights up. In this setup, only two sides can be observed simultaneously
(e.g., A and C, but not B). }
\end{figure}
where we can observe not only A and B, but also the top, C. The outcomes
of an observation will modeled by $\pm1$-valued random variables,
$\mathbf{A}$, $\mathbf{B}$, and $\mathbf{C}$, corresponding to
which side of the cube's face glows (marked in the figure with $+$
and $-$ ). It is clear that there is a one-one correspondence between
the region inside the cube and what values the random variables take
if the firefly blinks. For example, the faces of the cube divide it
naturally into octants, and if the firefly is in one octant, the value
of $\mathbf{A}$, $\mathbf{B}$, and $\mathbf{C}$ will be determined.
We can think of the firefly blinking in one octant as corresponding
to an elementary event in the sample space of a probability space
$\left(\Omega,\mathcal{F},p\right)$, and we can label them according
to the values of $\mathbf{A}$, $\mathbf{B}$, and $\mathbf{C}$,
i.e. $\Omega=\left\{ \omega_{abc},\omega_{\overline{a}bc},\omega_{a\overline{b}c},\omega_{ab\overline{c}},\omega_{a\overline{b}\overline{c}},\omega_{\overline{a}b\overline{c}},\omega_{\overline{a}\overline{b}c},\omega_{\overline{a}\overline{b}\overline{c}}\right\} $,
where we use the notation that the subscripts correspond to the outcome
of the random variables, with the barred ones being $-1$ and the
other $+1$ (e.g., $\omega_{a\overline{b}c}$ corresponds to the octant
where $\mathbf{A}=1$, $\mathbf{B}=-1$, and $\mathbf{C}=1$). 

Let us further assume that, like the two-sided box of Figure \ref{fig:Foulis's-box},
we cannot observe all three sides at the same time, but only two.
This means that we do not only have access to the values of $E\left(\mathbf{A}\right)$,
$E\left(\mathbf{B}\right)$, and $E\left(\mathbf{C}\right)$, but
also to their second moments, $E\left(\mathbf{AB}\right)$, $E\left(\mathbf{BC}\right)$,
and $E\left(\mathbf{AC}\right)$, (which together with the individual
expectations fully determine the joint distribution of each pair of
random variables). It is easy to see that if we start with the above
sample space, we can impose constraints on the values of the moments.
To see this, consider the following table:

\begin{center}
\begin{tabular}{|c|c|c|c|c|}
\hline 
 & $\mathbf{AB}$ & $\mathbf{AC}$ & $\mathbf{BC}$ & $\mathbf{AB}+\mathbf{AC}+\mathbf{BC}$\tabularnewline
\hline 
\hline 
$\omega_{abc}$ & $+1$ & $+1$ & $+1$ & $+3$\tabularnewline
\hline 
$\omega_{\overline{a}bc}$ & $-1$ & $-1$ & $+1$ & $-1$\tabularnewline
\hline 
$\omega_{a\overline{b}c}$ & $-1$ & $+1$ & $-1$ & $-1$\tabularnewline
\hline 
$\omega_{ab\overline{c}}$ & $+1$ & $-1$ & $-1$ & $-1$\tabularnewline
\hline 
$\omega_{a\overline{b}\overline{c}}$ & $-1$ & $-1$ & $+1$ & $-1$\tabularnewline
\hline 
$\omega_{\overline{a}b\overline{c}}$ & $-1$ & $+1$ & $-1$ & $-1$\tabularnewline
\hline 
$\omega_{\overline{a}\overline{b}c}$ & $+1$ & $-1$ & $-1$ & $-1$\tabularnewline
\hline 
$\omega_{\overline{a}\overline{b}\overline{c}}$ & $+1$ & $+1$ & $+1$ & $+3$\tabularnewline
\hline 
\end{tabular}
\par\end{center}

\noindent Given that the table holds for individual values, the expected
value for each of the columns\footnote{The quantity $\mathbf{AB}+\mathbf{AC}+\mathbf{BC}$ is itself a random
variable.} are simply a convex combination of their values (which weights for
this convex combination depends on the particular values of the observed
expectations). An immediate consequence is that for the probability
space given, the moments must be always such that 
\begin{equation}
-1\leq E\left(\mathbf{AB}\right)+E\left(\mathbf{AC}\right)+E\left(\mathbf{BC}\right),\label{eq:suppes-zanotti-partial}
\end{equation}
which is the right-hand-side of the Suppes-Zanotti inequalities \cite{suppes_when_1981}.
In other words, if there is a probability space that describes all
the moments for $\mathbf{A}$, $\mathbf{B}$, and $\mathbf{C}$, then
inequality (\ref{eq:suppes-zanotti-partial}) must be satisfied. 

Here we point out that violations of inequality \cite{suppes_when_1981}
correspond to violations of logical consistency, as indicated by Abramsky
and Hardy \cite{abramsky_logical_2012}. To violate \cite{suppes_when_1981},
we need in the convex combination of elements at least some events
that lead to values on the right column that are less than -1. One
such element, for example, is $E\left(\mathbf{AB}\right)=E\left(\mathbf{AC}\right)=E\left(\mathbf{BC}\right)=-1$,
which adds up to $-3$. For these moments, if $\mathbf{A}=1$, then
$E\left(\mathbf{AB}\right)$ implies $\mathbf{B}=-1$, and from $E\left(\mathbf{BC}\right)$
it follows that $\mathbf{C}=1$, which finally leads, from $E\left(\mathbf{AC}\right)$,
to $\mathbf{A}=-1$, a contradiction. The contradiction comes from
the assumption that, say, the random variable $\mathbf{A}$ in the
experiment that measures $E\left(\mathbf{AB}\right)$ is the same
as the ones in the experiment $E\left(\mathbf{AC}\right)$. However,
as we will see, this is not the case, and outcomes of experiments
can depend on contexts. 

To show this let us we tweak the firefly example. As we mentioned,
the box in Figure \ref{fig:Three-sided-Foulis-box} is designed such
that one can only observe two sides at a time. This could be done
by having some mechanism attached to the box that prevents the observer
to see what happens on one of the sides. Let us now connect the mechanism
that selects which sides we can observe to a biasing mechanism inside
the box. This biasing mechanism turns on (inside the box) little devices
that release at random times\footnote{But with expected time intervals that are of the same order of the
expected period in between blinks for the firefly.} pheromones that attract the firefly. If we place those pheromone-releasing
devices in the right place, we can rig the box such that we have higher
probabilities of finding the firefly only in certain octants. Furthermore,
by a careful choice of octants, we can have it built such that the
second moment of, say, $\mathbf{A}$ and $\mathbf{B}$, is close to
$-1$. Because the pheromone-releasing mechanism is connected to the
side-selection mechanism, we can also make it change when we decide
to observe $\mathbf{B}$ and $\mathbf{C}$ or $\mathbf{A}$ and $\mathbf{C}$,
such that their second moment is also $-1$. Of course, $E\left(\mathbf{AB}\right)=E\left(\mathbf{AC}\right)=E\left(\mathbf{BC}\right)\approx-1$
violates (\ref{eq:suppes-zanotti-partial}).

What is happening in the previous example is simple: inequalities
(\ref{eq:suppes-zanotti-partial}) are violated because each observational
setting, i.e. the decision of which two variables to observe, corresponds
to a different experimental condition. This is because the choice
of observing $\mathbf{A}$ and $\mathbf{B}$ instead of any other
pair changes the places where the pheromones are being released. In
other words, the probability space $\left(\Omega,\mathcal{F},p\right)$
assumes that the values of the random variables $\mathbf{A}$ in the
experiment with $\mathbf{B}$ are compatible with $\mathbf{A}$ in
the experiment with $\mathbf{C}$. But such compatibility is impossible.
To illustrate this in a different way, let us examine what happens
to the octants as we impose the moments $E\left(\mathbf{AB}\right)=E\left(\mathbf{AC}\right)=E\left(\mathbf{BC}\right)=-1$.
As we saw above, the sample space $\Omega$ is represented in terms
of the octants, one for each elementary event in $\Omega=\left\{ \omega_{abc},\omega_{\overline{a}bc},\omega_{a\overline{b}c},\omega_{ab\overline{c}},\omega_{a\overline{b}\overline{c}},\omega_{\overline{a}b\overline{c}},\omega_{\overline{a}\overline{b}c},\omega_{\overline{a}\overline{b}\overline{c}}\right\} $.
\begin{figure}
\centering{}\includegraphics[scale=0.5]{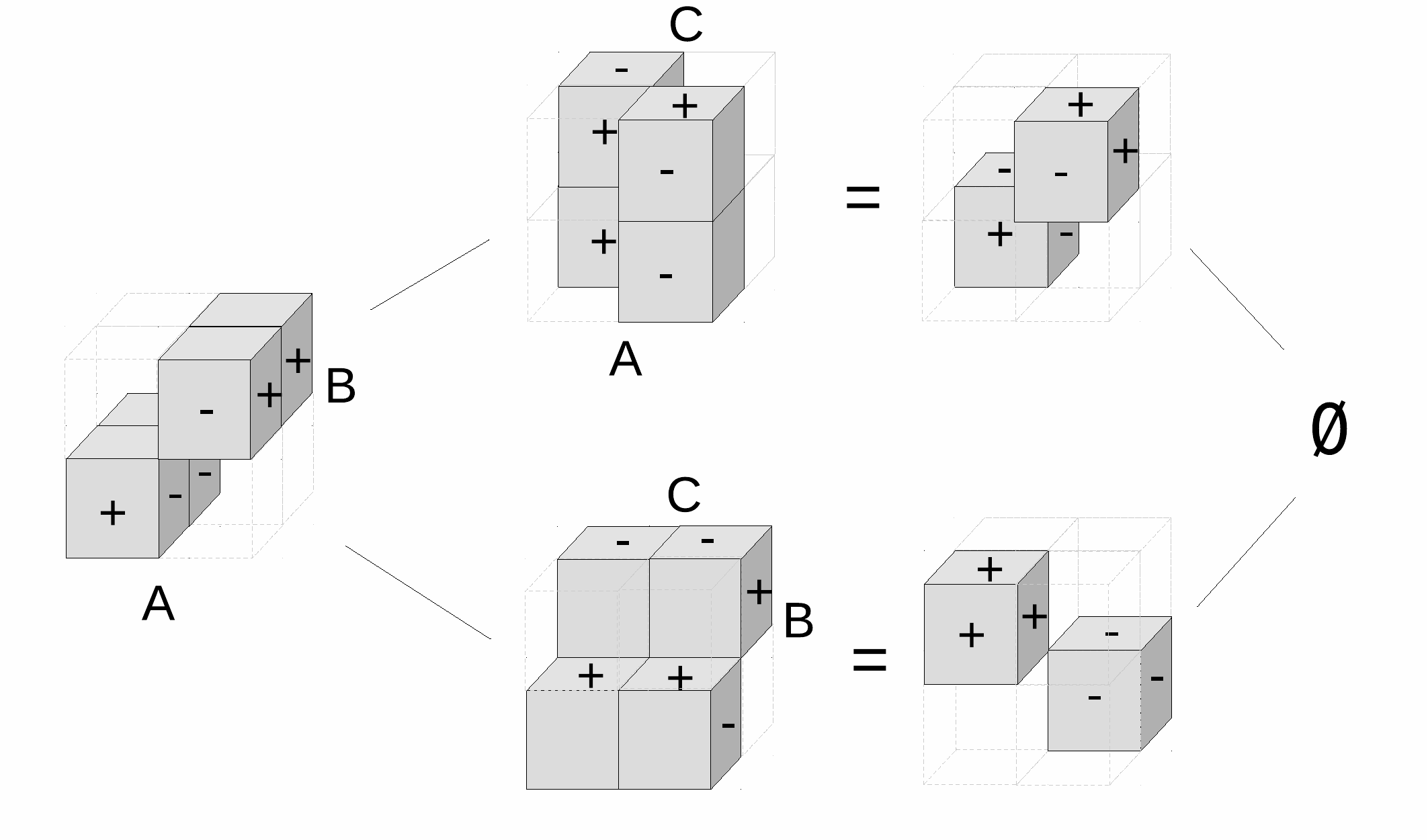}\protect\caption{\label{fig:Contextual-Foulis's-box}Contextual firefly box with two
observables at a time. On the left hand side we have the possible
octants where the firefly is when $\left(\mathbf{A},\mathbf{B}\right)$
are observed, and in the middle $\left(\mathbf{A},\mathbf{C}\right)$
and $\left(\mathbf{B},\mathbf{C}\right)$ (bottom). After the equal
signs we have the intersection between $\left(\mathbf{A},\mathbf{B}\right)$
and $\left(\mathbf{A},\mathbf{C}\right)$ (top right) and between
$\left(\mathbf{A},\mathbf{B}\right)$ and $\left(\mathbf{B},\mathbf{C}\right)$
(bottom right). It is clear that the intersection between those scenarios
lead to the empty set, an impossible event. }
\end{figure}
 To reproduce the $E\left(\mathbf{AB}\right)=-1$ observation (when
$\mathbf{C}$ is not observable), the firefly would need to be in
the two regions denoted by the two prisms on the left hand side of
Figure \ref{fig:Contextual-Foulis's-box}. The further constraint
that $E\left(\mathbf{AC}\right)=-1$ leads to a smaller region of
the sample space, corresponding to only two cubes (center top on Figure
\ref{fig:Contextual-Foulis's-box}). But that implies that $E\left(\mathbf{BC}\right)=1$,
and there are no points in the sample space that correspond to $E\left(\mathbf{AB}\right)=E\left(\mathbf{AC}\right)=E\left(\mathbf{BC}\right)=-1$
(this, by the way, is straightforward from our table above, and is
also shown pictorially in Figure \ref{fig:Contextual-Foulis's-box}).
Therefore, the regions where the firefly is depends on which sides
of the box you are observing. This characteristic is called \emph{contextuality}. 

Thus, contextuality, for us, can be stated in the following way. If
a set of random variables, measured under different experimental conditions
and never all at the same time, cannot be represented as partitions
of a joint probability distribution, then they are contextual. From
our example it should be clear that the nonexistence of a joint probability
distribution for contextual random variables was not based on taking
a coarse-grained probability space over the firefly's path. But to
make it explicit, we can notice that any probability space that reproduces
the outcomes of $\mathbf{A}$, $\mathbf{B}$, and $\mathbf{C}$ must
have as part of its algebra the elements $p_{abc},\ldots,p_{\overline{a}\overline{b}\overline{c}}$,
and therefore it cannot have a proper probability distribution over
it. 

Following the Contextuality-by-Default (CbD) approach \cite{dzhafarov_contextuality_2014-1,dzhafarov_contextuality_2015,kujala_necessary_2015},
one should assume that different experiments (e.g, observing $(A,B)$,
observing $(B,C)$, and observing $(A,C)$ in the above example) are
stochastically unrelated and therefore modeled on distinct probability
spaces. Indeed, only one experiment can be performed at a time so
there is no pairing-scheme to justify defining the random variables
of different experiments on the same probability space. 

Thus, indexing properties by subscripts $i=1,\dots,M$ and different
contexts by superscripts $j=1,\dots,N$, let us model the result of
observing property $i$ in context $j$ by the random variable $\mathbf{R}_{i}^{j}:\Omega_{j}\to E_{i}$,
where $\mathbf{R}_{i}^{j}$ with different $i$ but same $j$ are
all jointly distributed but $\mathbf{R}_{i}^{j}$ and $\mathbf{R}_{i'}^{j'}$
for $j\ne j'$ are stochastically unrelated for all $i,i'\in\{1,\dots,n\}$,
equal or not. For keeping the notation uncluttered, we assume that
when property $i$ does not appear in context $j$, the expression
$\mathbf{R}_{i}^{j}$ is undefined and left out from any enumerations.
Using this convention, we denote by $\mathbf{R}^{j}=\{\mathbf{R}_{i}^{j}:i=1,\dots,m\}$
the jointly distributed set of random variables modeling the measurement
of all properties appearing in context $j$. 
\begin{defn}
A collection of random variables $\mathbf{R}_{i}^{j}$ is said to
be \emph{consistently connected} if $\mathbf{R}_{i}^{j}\sim\mathbf{R}_{i}^{j'}$
for all $i\in\{1,\dots,m\}$ and all contexts $j,j'\in\{1,\dots,n\}$
in which the property $i$ appears (here we use the notation $\mathbf{A}\sim\mathbf{B}$
to signify that ``$\mathbf{A}$ has the same distribution as $\mathbf{B}$'').
If a system is not consistently connected, it is said to be \emph{inconsistently
connected.}
\end{defn}
Intuitively, \emph{consistently connected }means that one cannot find
differences between a random variable in one context and another by
solely observing this random variable\footnote{This, by the way, is related to the no-signaling condition in physics.}.
For example, if we observe $\mathbf{A}$ in the context of $\mathbf{B}$,
and we then observe $\mathbf{A}$ in the context of $\mathbf{C}$,
no scrutiny of the distribution or values of $\mathbf{A}$ can tell
us which context it was observed in if the random variables are consistently
connected. 

Random variables that are not consistently connected are obviously
context-dependent. As we will see in Section 4 below, many of the
examples in physics and psychology are context-dependent because of
being inconsistently connected. However, it is still possible for
a system of random variables to present contextuality even if they
are consistently connected, and we need to distinguish those cases.
This is the essence of the following definitions. 
\begin{defn}
A \emph{coupling} of random variables $\mathbf{X}_{1},\dots,\mathbf{X}_{n}$
(that may be defined on different probability spaces) is any jointly
distributed set of random variables $\mathbf{Z}_{1},\dots,\mathbf{Z}_{n}$
such that $\mathbf{X}_{1}\sim\mathbf{Z}_{1}$, $\dots$, $\mathbf{X}_{n}\sim\mathbf{Z}_{n}$.
\end{defn}
Intuitively, a coupling imposes a joint distribution on a set of random
variables and hence formalizes the concept of finding a common sample
space for a set of random variables. Thus, we can define the traditional
understanding of (non-)contextuality in a mathematically rigorous
form as follows.
\begin{defn}
\label{def:contextuality}A collection of random variables $\mathbf{R}_{i}^{j}$
is non-contextual if and only if there exists a coupling $\mathbf{Q}^{1},\dots,\mathbf{Q}^{n}$
of $\mathbf{R}^{1},\dots\mathbf{R}^{n}$ such that $\mathbf{Q}_{i}^{j}=\mathbf{Q}_{i}^{j'}$
for all $i\in\{1,\dots,m\}$ and all contexts $j,j'\in\{1,\dots,n\}$
in which the property $i$ appears. 
\end{defn}
Definition \ref{def:contextuality} only holds for consistently connected
systems, as it requires, as its consequence, that $\mathbf{R}_{i}^{j}\sim\mathbf{R}_{i}^{j'}$
for all $i\in\{1,\dots,m\}$ and all contexts $j,j'\in\{1,\dots,n\}$
in which property $i$ appears. If the coupling of Definition \ref{def:contextuality}
exists, we can denote $\mathbf{Q}_{i}=\mathbf{Q}_{i}^{j}=\mathbf{Q}_{i}^{j'}=\dots$
for all $i=1,\dots,m$ and all contexts $j,j',\dots\in\{1,\dots,n\}$
in which the property $i$ appears. We can think intuitively of the
distributions of $\mathbf{R}^{j}$'s as observable marginal distributions
of the hypothetical larger system $\mathbf{Q}_{1},\dots,\mathbf{Q}_{m}$.
Thus, the collection of r.v.'s is non-contextual if it is possible
to ``sew'' the observed marginal probabilities together to produce
a larger probability distribution over the whole set of properties
\cite{dzhafarov_all-possible-couplings_2013,dzhafarov_contextuality_2014,de_barros_unifying_2014,de_barros_negative_2015}.\emph{
}As mentioned, in Quantum Mechanics, and perhaps Psychology, it may
not be possible to do that, but in many cases the marginal probabilities
are compatible with a \emph{signed} joint probability distribution
of $\mathbf{Q}_{1},\dots,\mathbf{Q}_{m}$. 

We are now left with the following three situations for collections
of r.v.'s measured under different contexts: they are non-contextual
(i.e., they can be imposed on a proper joint probability distribution
in which r.v's representing the same property are always equal); the
random variables are contextual and consistently connected (in the
next section, we show that they can then be imposed on a signed joint
distribution); and they are inconsistently connected\footnote{Here we should add a comment on terminology. In the CbD approach,
the definition of noncontextuality is extended \cite{dzhafarov_contextuality_2015,dzhafarov_is_2015}
to inconsistently connected systems by allowing $\mathbf{Q}_{i}^{j}=\mathbf{Q}_{i}^{j'}=\dots$
to not hold as long as the probability of it holding is in a certain
well-defined sense maximal for each $i$. This allows one to detect
contextuality on top of inconsistent connectedness and so in the most
recent terminology of CbD, a system can be inconsistently connected
and yet not contextual. For the present paper, since we are mostly
focusing on the NP approach which does not apply to inconsistently
connected system, it suffices to use the traditional understanding
of contextuality.}. For inconsistently connected systems, things are a little more subtle,
and since NP cannot yet deal with them, we refer to the works of Dzhafarov
and Kujala \cite{kujala_necessary_2015,kujala_probabilistic_2015}.

\section{Negative Probabilities\label{sec:Negative-Probabilities}}

As we mentioned in the previous section, contextuality appears not
only in psychology but also in physical systems. In this Section we
describe one possible approach to describing contextuality: Negative
Probabilities (NP). NP are a generalization of standard Kolmogorovian
probabilities to accommodate the observations performed in a system
that behaves in a contextual way. 

It is important at this point to understand why changing the theory
of probability is a desirable approach. When we have context-dependent
r.v.'s, such as the $\mathbf{A}$, $\mathbf{B}$, and $\mathbf{C}$
in the firefly box, one could argue that the nonexistence of a joint
probability distribution comes from a mistake: the identification
of a random variable, say $\mathbf{A}$, in two different contexts
(e.g. $\left(\mathbf{A},\mathbf{B}\right)$ and $\left(\mathbf{A},\mathbf{C}\right)$
as being the same). This is clearly what is happening, and the solution
to it is, following Dzhafarov and Kujala's Contextuality-by-Default
approach, to clearly label each r.v. according to its context. However,
there are cases when such distinction may not highlight important
non-trivial features of a system. One such case is the famous Bell-EPR
experiment, which we describe below. For this experiment, because
the experiments that measure, e.g., $\mathbf{A}$ and $\mathbf{B}$
should not interfere with each other, for physical reasons, it makes
no sense to label them differently. However, the Bell-EPR system is
contextual, and using the same label brings this contextuality to
the surface in a very dramatic way. Therefore an extended probability
theory may help shed light in some of those contextual cases. 

Because contextuality is equivalent to the non-existence of a joint
probability distribution (see Proposition \ref{prop:no-signaling}
below) for a collection of random variables, some proposals for dealing
with contextual systems are to simply change the theory of probability.
This is what was done in quantum mechanics, where the complementarity
principle, whereupon some variables were forbidden in principle to
be observed simultaneously, opened up the need to describe such contextual
systems with the formalism of measures over Hilbert spaces. However,
a question in QM is what are the principles behind such a specific
generalized probability theory? Why do we use Hilbert spaces? These
questions form an important topic of research, and are yet unanswered.
Similarly, these questions can also be asked for psychology. Why is
the quantum formalism adequate to model psychological systems? Up
to now, it seems that all arguments about using the quantum formalism
are related to specific examples that form a subset of those in physics,
since they all involve inconsistently connected systems.

So, instead of using quantum probabilities, as do researchers in quantum
cognition, we propose a more general framework given by negative probabilities.
Our definition of NP is a straightforward generalization of Kolmogorov's
probability. A discrete signed probability space is given by a triple
$(\Omega,\mathcal{F},p)$ with the same components as a proper probability
space, except that the function $p:\Omega\to\mathbb{R}$ is allowed
to attain negative values and values larger than 1, as long as it
still satisfies $p(\Omega)=1$. The probability of an event $E\in\mathcal{F}$
is still calculated as $p(S)=\sum_{\omega\in S}p(\{\omega\})$, like
in proper probability spaces, and random variables and expectations
are defined analogously to those of proper probability spaces.

Let us motivate the above definition. First, we notice that, in Definition
\ref{def:contextuality}, contextuality was defined as the impossibility
to impose a joint distribution on $\mathbf{R}^{1},\dots,\mathbf{R}^{n}$
such that the stochastically unrelated random variables representing
the same property are always equal in this joint. However, for consistently
connected systems, it turns out it is always possible to find such
a joint on a signed probability space:
\begin{prop}
\label{prop:no-signaling}For a collection $\mathbf{R}_{i}^{j}$ of
discrete r.v.'s on finite sample spaces, the following are equivalent
\begin{enumerate}
\item there exists on a signed probability space jointly distributed r.v.'s
$\mathbf{Q}_{1},\dots,\mathbf{Q}_{m}$ such that for all $j=1,\dots,n$,
it holds $\{\mathbf{R}_{i}^{j},\mathbf{R}_{i'}^{j},\dots\}\sim\{\mathbf{Q}_{i},\mathbf{Q}_{i'},\dots\}$
where $i,i',\dots$are the properties appearing in context $j$. (Here
``$\sim$'' is taken to refer to the joint distributions of the
two sets of r.v.'s).
\item the collection $\mathbf{R}_{i}^{j}$ is consistently connected.
\end{enumerate}
\end{prop}
\begin{proof}
See \cite{abramsky_sheaf-theoretic_2011,al-safi_simulating_2013,oas_exploring_2014}.\end{proof}
\begin{defn}
Let $\mathbf{R}_{i}^{j}$ be a consistently connected collection of
r.v.'s. Then, the \emph{minimum L1 probability norm}, denoted\emph{
$M^{*}$, }or simply \emph{minimum probability norm}, is given by
$M^{*}=\min\sum_{\omega\in\Omega}\left|p\left(\left\{ \omega\right\} \right)\right|$,
where the minimization is over all signed probability spaces $(\Omega,\mathcal{F},p)$
and r.v.'s $\mathbf{Q}_{1},\dots,\mathbf{Q}_{n}$ on it that satisfy
condition 1 of Proposition \ref{prop:no-signaling}.
\end{defn}

From this definition it is easy to prove the following: 
\begin{prop}
\label{prop:m=00003D1impliesjoint}A consistently connected collection
of r.v.'s is non-contextual if and only if $M^{*}=1$. \end{prop}
\begin{proof}
See \cite{de_barros_negative_2015}.
\end{proof}
If follows that since $M^{*}$ can be greater than one for contextual
systems, and that the greater the value of $M^{*}$ the further away
from a proper probability distribution it lies (due to the strong
relationships imposed, e.g., by the moments of the random variables),
it is natural to interpret $M^{*}$ as a measure of contextuality:
the larger the value of $M^{*}$, the more contextual the system \cite{de_barros_unifying_2014}. 

From Proposition \ref{prop:no-signaling}, it follows that inconsistently
connected systems of random variables cannot be described with negative
probabilities. Here we are left with only one possibility. If, for
a system of random variables, some of them are not consistently connected,
then we need to face the fact that they are not the same random variable,
and label them accordingly, following the prescription of Contextuality-by-Default
\cite{dzhafarov_contextuality_2014-1,dzhafarov_qualified_2014}. 

We end this Section with some comments about the meaning of NP. One
of the main obstacles to the use of NP is the lack of an interpretation.
After all, what meaning should we give to them? If probabilities are,
as in some objective views, given by relative frequencies of actual
realizable events, how can we even consider a probability to be negative? 

First, we should point out that NP are not directly observable, but
only inferrable. For instance, in the firefly box, negative probabilities
appear exactly because we cannot observe all three sides simultaneously
(or know where the firefly is). Were we able to observe all three
simultaneously, then a joint probability distribution would necessarily
exist. With the observable moments, any attempt to create probabilities
that have marginals consistent with the moments lead to NP. 

Our second point is that NP may be useful in certain applications.
For example, in physics an important question is what are the physical
principles that define QM. NP may be an adequate tool to help us understand
those principles. We will not explore this application of NP here,
but the interested reader is referred to \cite{oas_exploring_2014,oas_survey_2015}.

That said, there are ways to interpret NP, even consistently with
a frequentist interpretation. Here we will briefly sketch how some
of those interpretations work, but the interested reader should refer
to the cited references. We start with Andrei Khrennikov's $p$-adic
interpretation. Khrennikov \cite{khrennikov_p-adic_1993,khrennikov_p-adic_1993-1,khrennikov_statistical_1993,khrennikov_discrete_1994,khrennikov_p-adic_1994,khrennikov_interpretations_2009}
showed that for the frequentist interpretation proposed by von Mises,
where probabilities are defined as the convergent ratio of infinite
sequences, NP appear in sequences where the usual Archimedian metric
does not converge to a specific value (i.e., sequences not satisfying
the principle of stabilization). When Archimedian metrics do not converge,
$p$-adic metrics may do so, and in those cases NP appear as the $p$-adic
limiting case. In other words, Khrennikov interprets infinite sequences
that do not satisfy the principle of stabilization as arising from
contextuality, and describable by negative probabilities. However,
the relationship between NP and observations, in this interpretation,
is not straightforward, as it depends on the particular $p$-adic
metric chosen. 

Another interpretation of NP, also frequentist, is the one proposed
by Abramsky and Brandenburger \cite{abramsky_sheaf-theoretic_2011,abramsky_operational_2014}.
They use, in the context of sheaf theory, the concept that events
may have two different types that may annihilate each other. In most
circumstances, when quantities are observed, no events are annihilated;
however, when there are context-dependent observables, they are context-dependent
because each context determines a different interaction between the
observables through their annihilation. 

Finally there is Szekely's ``half-coin'' interpretation of NP \cite{ruzsa_convolution_1983,szekely_half_2005}.
The idea is that two probability distributions that are negative may
give rise to a non-negative proper probability distribution. In this
interpretation, negative probabilities $P$ are related to a proper
probability $p$ via a convolution equation $P*p_{-}=p_{+}$, which
is always possible to be found \cite{ruzsa_convolution_1983,szekely_half_2005}.
This convolution means that for a random variable $\mathbf{X}$ whose
(negative) probability distribution is $P$, there exists two other
random variables, $\mathbf{X}_{+}$ and $\mathbf{X}_{-}$ with proper
probability distributions ($p_{+}$ and $p_{-}$, respectively) and
such that $\mathbf{X}=\mathbf{X}_{+}-\mathbf{X}_{-}$. As one can
see, this interpretation is closely related to that of Abramsky and
Brandenburger. 

In this paper we favor a more pragmatic ``interpretation.'' Negative
probabilities are taken here to be simply an accounting tool, one
that provides us the best subjective information about systems which
do not have an objective probability distribution, as it is the closest
distribution to a proper one (via normalization of the L1 norm). This
is analogous to the use of negative numbers in mathematics, which
was considered by many absurd. For example, the famous mathematician
Augustus De Morgan wrote the following about negative numbers \cite[pg. 72]{morgan_study_1910}. 
\begin{quote}
``Above all, he {[}the student{]} must reject the definition still
sometimes given of the quantity $-a$, that it is less than nothing.
It is astonishing that the human intellect should ever have tolerated
such an absurdity as the idea of a quantity less than nothing; above
all, that the notion should have outlived the belief in judicial astrology
and the existence of witches, either of which is ten thousand times
more possible.'' 
\end{quote}
However, nowadays we understand that negative numbers can be a useful
bookkeeping device. For example, when tracking a store inventory,
one would not be overly concerned about something such as ``$-30$
rolls of toilet paper'' in our spreadsheet, and equate such a line
to ``the existence of witches.'' We approach NP the same way, asking
whether it can be a useful device that may not only help us in computations
but also give us further insights in some situations, as mentioned
above. But, as De Morgan, we consider a statement such as ``event
A has probability $-0.1$'' on equal terms with judicial astrology.

\section{Some examples and applications\label{sec:Some-examples}}

Let us now examine some examples of contextual systems, and how they
can be described (or not) with negative probabilities. We already
gave an example of a contextual system above, with the $\mathbf{A}$,
$\mathbf{B}$, $\mathbf{C}$ random variables from the three-sided
firefly box. Here we will look at examples from physics, in particular
Quantum Mechanics, and then move to psychology\footnote{For some of our physics examples, we assume that the reader is familiar
with the mathematical formalism of QM. Readers not familiar with it
may wish to skip the details, since they do not affect the overall
understanding of this paper, or may refer to the many available texts
on this subject (e.g. \cite{cohen-tannoudji_quantum_1977} or \cite{peres_quantum_1995}). }. 

Perhaps the most important example of contextuality in physics is
the double-slit experiment, as it contains in its essence the complementarity
principle. So, here we start this section with this experiment, but
in a simplified version given by the Mach-Zehnder interferometer,
which captures the main features of complementarity. We then re-examine
the firefly box in more details, showing how negative probabilities
can model some of its outcomes. Next, we review the first example
where contextuality was recognized as playing a key role in Quantum
Mechanics, the famous Kochen-Specker theorem \cite{kochen_problem_1967}.
Finally, as a last physics example, we investigate the Bell-EPR with
negative probabilities. We then move to the contextual cases in quantum
cognition, and we discuss how those are related to the different physics
cases shown before. We end this section with a discussion of negative
probabilities as a possible way to measure the contextuality of an
observable system.

\subsection{Interference experiments\label{sub:Interference-experiments}}

In the double-slit experiment, a particle impinges on a solid barrier
that has on it two small and parallel slits close to each other. The
particle has a probability of passing through the slits, later on
being detected on a scintillating screen. Contextuality in this experiment
appears as a manifestation of the wave/particle duality: the places
where the screen scintillates depend on whether we know any which-path
information for the particle, i.e., whether the particle went through
one slit and not the other (see \cite{feynman_feynman_2011} for a
detailed discussion of the double slit). 

Since the detection rates of an observed event depend on the context,
it is immediate that the double-slit experiment exhibits trivial contextuality.
It thus follows that it cannot be described using negative probabilities\footnote{Except if one makes special counterfactual assumptions, as common
in certain physics experiments \cite{de_barros_negative_2014,de_barros_negative_2015}.}. However, due to its importance in many applications of the quantum
formalism to psychology, we present a brief discussion of it here
in a simplified form. To do so, we use a conceptually similar setup
where the slits and the screen are replaced by the Mach-Zehnder interferometer
shown in Figure (see Figure \ref{fig:Mach-Zehnder-Interferometer};
for a more detailed discussion of the MZI, see references \cite{de_barros_negative_2015,de_barros_examples_2015}).
\begin{figure}
\begin{centering}
\includegraphics[scale=0.45]{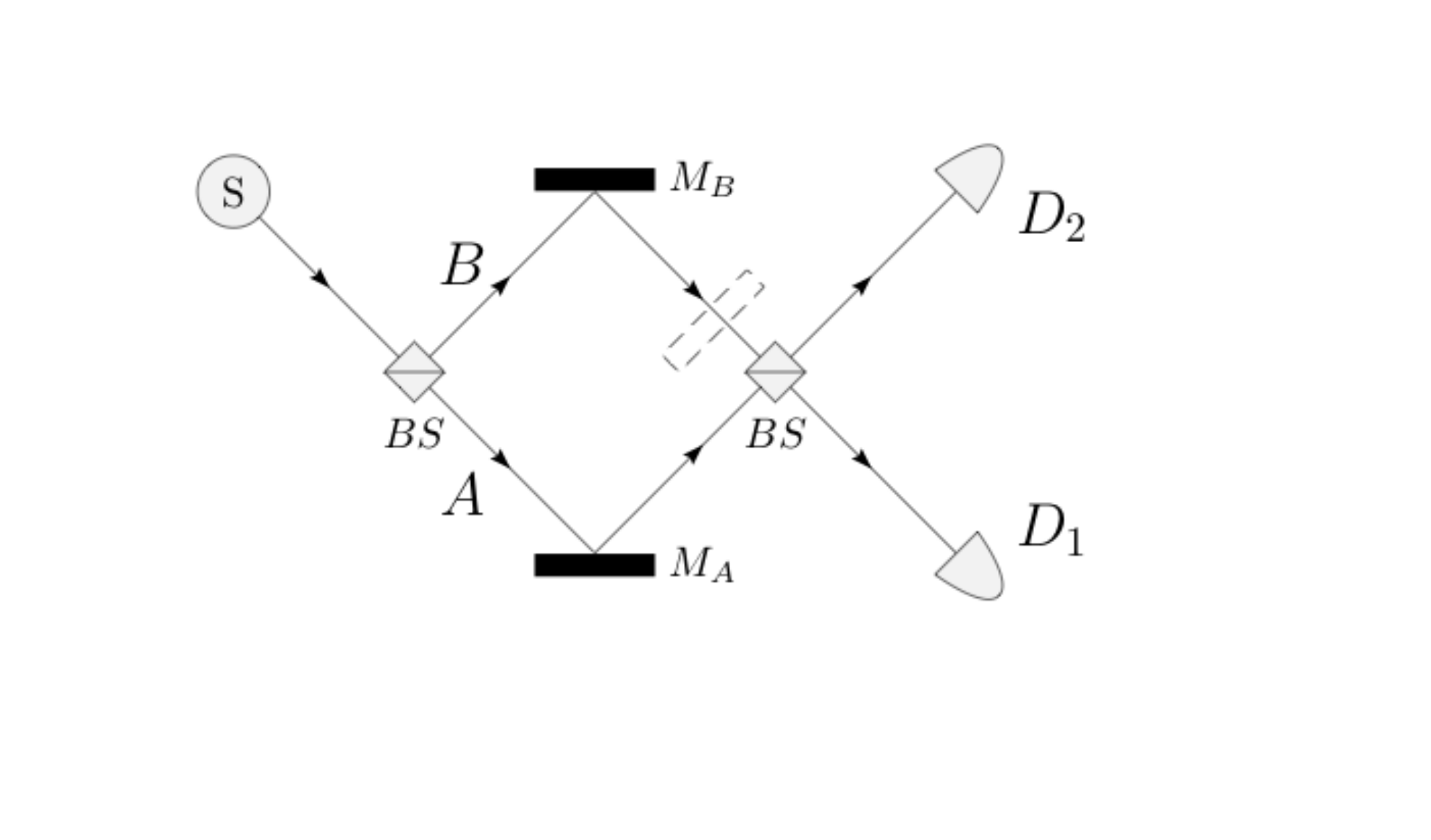}
\par\end{centering}

\protect\caption{\label{fig:Mach-Zehnder-Interferometer}Mach Zehnder Interferometer
(MZI). A source $S$ sends a particle beam that impinges on the first
beam splitter $BS$. The beam is them divided by $BS$ into equal-intensity
(i.e. particle numbers, on average) beams that travel to both arms
(paths) $A$ and $B$ of the interferometer, reflecting on surfaces
$M_{A}$ and $M_{B}$. The beams from arms $A$ and $B$ are then
recombined in the second beam splitter. The outcomes are the two beams
detected at $D_{1}$ and $D_{2}$. }
\end{figure}
 Beam splitters send the particles into two directions in a random
way, such that if we place a particle detector after each of the outputs
of the beam splitter, we will see that in the long run the number
of particles going to one side approaches that of going to the other
side. After the first beam splitter, some reflecting surfaces redirect
the beams to another beam splitter, and the beams are recombined.
From QM, the whole system can be described mathematically by a wavefunction,
and the recombination of the two beams in the second beam splitter
leads to an interference effect. Careful positioning of the beam splitters
and reflecting surfaces allow for perfect interference, namely all
particles reaching $D_{1}$ and none reaching $D_{2}$. 

What makes the MZI interesting is that the placement of a detector
on \emph{either} path $A$ or path $B$ causes a collapse of the wavefunction,
thus changing the outcomes of a measurement of $D_{1}$ and $D_{2}$:
they now have the same probability of detecting a particle. However,
let us recall that if no detectors are placed on $A$ or $B$, the
particle has zero probability of reaching $D_{2}$. Furthermore, if
we simply block one of the paths, say by putting a barrier in $A$,
half of the particles going through $B$ will reach $D_{2}$. This
is seemingly disturbing, for how can we \emph{increase} the probability
of detection of $D_{2}$ when we actually \emph{decrease }the number
of ways in which the particle can reach $D_{2}$? This is the main
difficulty of the double-slit experiment. 

To see that the observations of $D_{1}$ and $D_{2}$ are contextual,
let $\mathbf{P}$ and $\mathbf{D}$ be two $\pm1$-valued random variables
representing which-path information and detection: $\mathbf{P}=1$
if the particle is detected on $A$ and $\mathbf{P}=-1$ otherwise,
$\mathbf{D}=1$ if the particle is detected in $D_{1}$ and $\mathbf{D}=-1$
otherwise. The MZI has two contexts: there is a detector on $A$ or
$B$, providing which-path information, or no detector. $\mathbf{D}$
measured under the no-which-path context has expected value $E\left(\mathbf{D}\right)=1$,
whereas a joint measurement of $\mathbf{D}$ and $\mathbf{P}$ gives
as marginal expectation the result $E\left(\mathbf{D}\right)=0$.
Thus, according to the above definition, $\mathbf{D}$ is inconsistently
connected. 

So, since it is inconsistently connected, how would we model the MZI
with NP? The fact that $\mathbf{D}$ changes when measured with $\mathbf{P}$
or not leads to the necessity of defining two different random variables,
$\mathbf{D}$ and $\mathbf{D}_{\mathbf{P}}$, where $\mathbf{D}_{\mathbf{P}}$
is simply the representation of the detectors under the which-path
information context. Clearly we can always write down a joint probability
distribution for $\mathbf{P}$, $\mathbf{D}$, and $\mathbf{D}_{\mathbf{P}}$,
but then there is no contextuality in this system, and no need for
negative probabilities.

\subsection{Three-sided Firefly Box\label{sub:Three-sided-Firefly-Box}}

For the three-sided firefly box of Figure \ref{fig:Three-sided-Foulis-box},
Suppes and Zanotti proved a necessary and sufficient condition for
the existence of a joint probability distribution, namely that $\mathbf{A}$,
$\mathbf{B}$, and $\mathbf{C}$ need to satisfy the following inequalities:
\begin{eqnarray}
-1 & \leq & E\left(\mathbf{AB}\right)+E\left(\mathbf{BC}\right)+E\left(\mathbf{AC}\right)\label{eq:suppes-zanotti}\\
 & \leq & 1+2\min\left\{ E\left(\mathbf{AB}\right),E\left(\mathbf{BC}\right),E\left(\mathbf{AC}\right)\right\} .\nonumber 
\end{eqnarray}
This example actually shows up in physics, and a weaker form of
inequalities (\ref{eq:suppes-zanotti}) are known in the physics literature
as the Leggett-Garg inequalities. If we restrict our variables to
consistently connected systems, it is straightforward to compute a
(negative) joint probability distribution consistent with expectations
violating (\ref{eq:suppes-zanotti}). In fact, imagine we have the
following moments
\begin{eqnarray*}
E\mathbf{\left(AB\right)} & = & \epsilon_{1},\\
E\mathbf{\left(BC\right)} & = & \epsilon_{2},\\
E\mathbf{\left(AC\right)} & = & \epsilon_{3}.
\end{eqnarray*}
To make the computations simpler, let us also assume that $E\left(\mathbf{A}\right)=E\left(\mathbf{B}\right)=E\left(\mathbf{C}\right)=0$.
Then, we can construct a (negative) probability space $\left(\Omega,\mathcal{F},P\right)$
with $\Omega=\left\{ \omega_{abc},\omega_{\overline{a}bc},\omega_{a\overline{b}c},\omega_{ab\overline{c}},\omega_{a\overline{b}\overline{c}},\omega_{\overline{a}b\overline{c}},\omega_{\overline{a}\overline{b}c},\omega_{\overline{a}\overline{b}\overline{c}}\right\} $,
and a $P$ satisfying the above marginals is given by 
\begin{eqnarray*}
p\left(\omega_{abc}\right) & = & \frac{1}{4}\left(1+\epsilon_{1}+\epsilon_{2}+\epsilon_{3}\right)-\alpha,\\
p\left(\omega_{\overline{a}bc}\right) & = & \frac{1}{4}\left(\alpha-\epsilon_{1}-\epsilon_{2}\right),\\
p\left(\omega_{a\overline{b}c}\right) & = & \frac{1}{8}\left(-1+\epsilon_{1}-\epsilon_{2}+\epsilon_{3}\right),\\
p\left(\omega_{ab\overline{c}}\right) & = & \frac{1}{8}\left(1+\epsilon_{1}-\epsilon_{2}-\epsilon_{3}\right),\\
p\left(\omega_{a\overline{b}\overline{c}}\right) & = & \frac{1}{4}\left(1+\epsilon_{3}\right)-\alpha,\\
p\left(\omega_{\overline{a}b\overline{c}}\right) & = & \frac{1}{8}\left(1-\epsilon_{1}+\epsilon_{2}-\epsilon_{3}\right),\\
p\left(\omega_{\overline{a}\overline{b}c}\right) & = & \frac{1}{8}\left(1+\epsilon_{1}-\epsilon_{2}-\epsilon_{3}\right),\\
p\left(\omega_{\overline{a}\overline{b}\overline{c}}\right) & = & \alpha,
\end{eqnarray*}
where $\alpha$ is a free parameter that takes a range of values given
by the moments $\epsilon_{1}$, $\epsilon_{2}$, and $\epsilon_{3}$
and by the minimization of the L1 norm. Notice that if the moments
violate (\ref{eq:suppes-zanotti}), then some of the probabilities
above will be negative, regardless of the values of $\alpha$, as
we should expect. 

To see what further information negative probabilities may provide,
we follow an example from \cite{de_barros_decision_2014,de_barros_quantum_2015}.
Imagine a decision-maker, Deana, who wants to invest in stocks. She
considers three companies, A, B, and C, about which she knows nothing.
In a wise move, Deana hires three ``experts,'' Alice, Bob, and Carlos,
to provide her with information about the companies. However, each
expert is specialized only in two of the companies, but not in all
(e.g. Alice knows a lot about A and B, but nothing about C). Imagine
now that the $\pm1$-valued random variables, $\mathbf{A}$, $\mathbf{B}$,
and $\mathbf{C}$, are supposed to model the experts' beliefs of a
stock value going up if $+1$ and down if $-1$ whenever asked about
it. We assume that our experts' opinions about each company A, B,
or C are consistently connected, i.e. they all agree about the expectations
of $\mathbf{A}$, $\mathbf{B}$, and $\mathbf{C}$. To make it simple
for our toy example, we set 
\begin{equation}
E\left(\mathbf{A}\right)=E\left(\mathbf{B}\right)=E\left(\mathbf{C}\right)=0.\label{eq:exp-xyz}
\end{equation}
Since Alice only knows about A and B, she can add to (\ref{eq:exp-xyz})
information about the second moment, and she claims 
\begin{equation}
E_{A}\left(\mathbf{AB}\right)=-1,\label{eq:xy}
\end{equation}
where we use the subscript $A$ to remind us that (\ref{eq:xy}) corresponds
to Alice's subjective belief\footnote{Our example is not easily translatable into objective probabilities,
but one could devise a situation where certain biases on the experts
sides could increase their assessment of second moments, thus recreating
the moments we use. }. Equation (\ref{eq:xy}) has a simple interpretation: Alice believes
that if the value of A goes up, B will certainly go down and vice
versa. Bob's and Carlos's beliefs are that 
\begin{equation}
E_{B}\left(\mathbf{AC}\right)=-\frac{1}{2},\label{eq:xz}
\end{equation}
and 
\begin{equation}
E_{C}\left(\mathbf{BC}\right)=0.\label{eq:yz}
\end{equation}
It is easy to see from (\ref{eq:suppes-zanotti}) that (\ref{eq:xy})--(\ref{eq:yz})
do not have a proper joint probability distribution. However, because
the random variables are consistently connected, there exists a negative
probability distribution consistent with (\ref{eq:exp-xyz})--(\ref{eq:yz}).

What can Deana do with her inconsistent expert information? The only
unknown to Deanna, in a certain sense, is the triple moment. The minimization
of the L1 norm provides a range of possible values for the triple
moments, namely, for the above expectations, 
\[
-\frac{1}{2}\leq E\left(\mathbf{XYZ}\right)\leq\frac{1}{2}.
\]
So, NP provide a range of possible values for the triple moment that
could be thought as the most reasonable range, given that the minimization
of L1 puts the negative measure as close to a proper probability distribution
as possible.

\subsection{Kochen-Specker Theorem\label{sub:Kochen-Specker-Theorem}}

Very early on, a heated discussion in the foundations of quantum mechanics
was whether the process of measuring a quantum system revealed the
actual value of a property or created it. To answer this question,
Kochen and Specker \cite{kochen_problem_1967} asked whether it was
possible to assign values $0$ or $1$ to a set of quantum properties
(corresponding to projection operators, the quantum equivalent to
yes/no measurements). If measurements revealed a property, then this
assignment of 0 and 1 values should be possible, but Kochen and Specker
showed this was not the case. To do so, they used $117$ projection
operators (projectors). However, a simpler proof with only $18$ projectors
in a four dimensional Hilbert space exists, and that form is followed
here. \cite{cabello_bell-kochen-specker_1996}. Let $P_{i}$ be a
collection of projectors, and let $\mathbf{V}_{i}$ be $\pm1$-valued
random variables taking values $-1$ or $+1$ depending on whether
the property $P_{i}$ is false or true, respectively. Since $P_{i}$
is determined uniquely by a vector in the Hilbert space, we use this
vector as the index $i$ for the projector and the random variable.
Consider the following set of equations, guaranteed to be satisfied
by the algebra of the chosen projection operators. 

\begin{align}
\mathbf{V}_{0,0,0,1}\mathbf{V}_{0,0,1,0}\mathbf{V}_{1,1,0,0}\mathbf{V}_{1,-1,0,0} & =-1,\label{eq:cabello-1st}\\
\mathbf{V}_{0,0,0,1}\mathbf{V}_{0,1,0,0}\mathbf{V}_{1,0,1,0}\mathbf{V}_{1,0,-1,0} & =-1,\\
\mathbf{V}_{1,-1,1,-1}\mathbf{V}_{1,-1,-1,1}\mathbf{V}_{1,1,0,0}\mathbf{V}_{0,0,1,1} & =-1,\\
\mathbf{V}_{1,-1,1,-1}\mathbf{V}_{1,1,1,1}\mathbf{V}_{1,0,-1,0}\mathbf{V}_{0,1,0,-1} & =-1,\\
\mathbf{V}_{0,0,1,0}\mathbf{V}_{0,1,0,0}\mathbf{V}_{1,0,0,1}\mathbf{V}_{1,0,0,-1} & =-1,\\
\mathbf{V}_{1,-1,-1,1}\mathbf{V}_{1,1,1,1}\mathbf{V}_{1,0,0,-1}\mathbf{V}_{0,1,-1,0} & =-1,\\
\mathbf{V}_{1,1,-1,1}\mathbf{V}_{1,1,1,-1}\mathbf{V}_{1,-1,0,0}\mathbf{V}_{0,0,1,1} & =-1,\\
\mathbf{V}_{1,1,-1,1}\mathbf{V}_{-1,1,1,1}\mathbf{V}_{1,0,1,0}\mathbf{V}_{0,1,0,-1} & =-1,\\
\mathbf{V}_{1,1,1,-1}\mathbf{V}_{-1,1,1,1}\mathbf{V}_{1,0,0,1}\mathbf{V}_{0,1,-1,0} & =-1.\label{eq:cabello-last}
\end{align}
A quick examination will reveal that the r.v.'s on each line correspond
to a set of commuting projectors. Because the $P_{i}$ in each line
are orthogonal, only one of the $\mathbf{V}_{i}$'s in each line can
be true at a time, and therefore the product of them must be $-1$.
The commutation of observables for each line means that each corresponding
random variable can be measured simultaneously, though this is not
true for all random variables in different lines. We can think of
each line as representing a particular context for the experiment. 

We can multiply the left hand side of (\ref{eq:cabello-1st})--(\ref{eq:cabello-last}),
and because each variable appears twice, their product must be one
(since $\mathbf{V}_{i}^{2}=1$, because it is a $\pm1$-valued random
variable). However, if we multiply the right hand side of (\ref{eq:cabello-1st})--(\ref{eq:cabello-last}),
their product is $-1$, and we reach a contradiction. The contradiction
comes from assuming that the random variable (say, $\mathbf{V}_{0,0,0,1}$)
in one experimental context (i.e., measured with $\mathbf{V}_{0,0,1,0}$,
$\mathbf{V}_{1,1,0,0}$, $\mathbf{V}_{1,-1,0,0}$) is the same as
the random variable in a different context (i.e., $\mathbf{V}_{0,0,0,1}$
in the context $\mathbf{V}_{0,1,0,0}$, $\mathbf{V}_{1,0,1,0}$, $\mathbf{V}_{1,0,-1,0}$).
Since each of the value combinations for the $\mathbf{V}_{i}$'s correspond
to an $\omega$ in a (course-grained) probability space, it follows
that there is no joint probability distribution underlying it. Therefore
the algebra of observables in Quantum Mechanics is contextual.

It is worth mentioning that the lack of a joint probability for the
above example is a consequence of the algebra of observables being
state independent. What this means is that for any system describable
by a four-dimensional Hilbert space we will reach the above contradiction,
regardless of how this system was initially prepared. Assuming consistent
connectedness (i.e. that the marginal expectations $\left\langle \mathbf{V}_{i}\right\rangle $
match between contexts), it is possible to find a negative probability
distribution that describes this system. However, such distributions
are quite large, consisting of signed probabilities for $2^{18}=262,144$
elementary events.

\subsection{Bell-EPR non-local contextuality}

Perhaps the most celebrated example of contextuality in QM is the
Bell-EPR thought experiment, which we present here in terms of random
variables. In this experiment, two spin-$1/2$ particles A and B are
emitted by a source and go to opposite directions, where Alice and
Bob measure them (see Figure \ref{fig:Bell-EPR-experiment}).
\begin{figure}
\centering{}\includegraphics[scale=0.5]{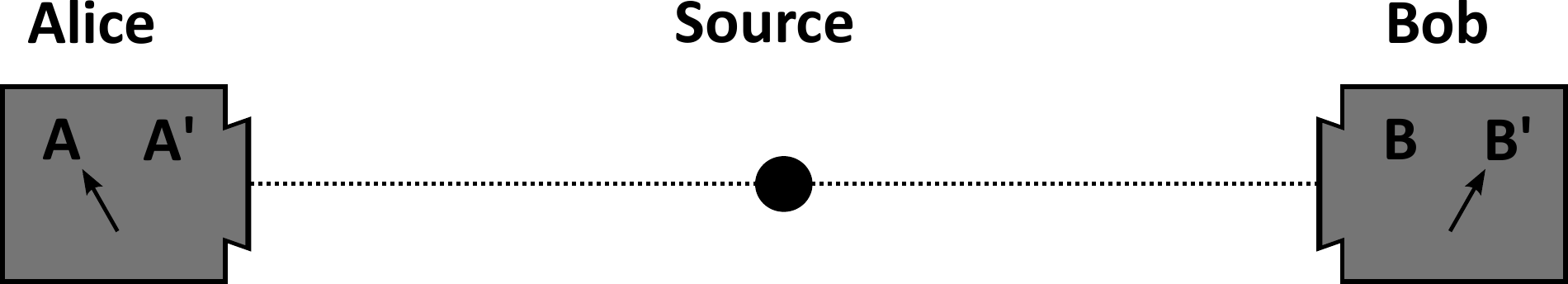}\protect\caption{\label{fig:Bell-EPR-experiment}Bell-EPR experiment. A source emits
two photons, one toward Alice's lab and another toward Bob's. Each
experimenter can make a decision on which direction of spin to measure,
represented in the figure by the settings $A$ and $A'$ for Alice
and $B$ and $B'$ for Bob. Outcomes of measurements are $\pm1$,
with equal probabilities. }
\end{figure}
 One of the possible states that can be prepared for such a source
is 
\begin{equation}
|\psi\rangle=\frac{1}{\sqrt{2}}\left(|+-\rangle-|-+\rangle\right),\label{eq:psi-EPR}
\end{equation}
where $|+-\rangle$ ($|-+\rangle$) corresponds to A having spin polarization
``$+1$'' (``$-1$'') and B ``$-1$'' (``$+1$'') in the $\mathbf{z}$
direction\footnote{The choice of the $\mathbf{z}$ direction is arbitrary. For simplicity
we use units where $\hbar/2=1$.}. It is clear that their spin $\mathbf{z}$ is negatively correlated,
as a ``$+1$'' or ``$-1$'' outcomes for particle A will result
in the same for particle B. Therefore, if Alice measures A's spin
in the $\mathbf{z}$ direction, then Bob's measurement in the same
direction is moot: his experimental outcomes are already determined
by Alice's. Einstein, Podolsky, and Rosen used this example to argue
that QM is incomplete: if we can know something (they called it an
``element of reality'') about particle B without affecting it (since
we measured A, whose measurement may be separated by spacelike interval
from Bob's own measurement), then the assumption in QM that the vector
(\ref{eq:psi-EPR}) is a complete description of its physical state
is incorrect \cite{einstein_can_1935}. 

We can argue that there is still no mystery with QM up to now, but
just an argument that QM should be incomplete. Einstein's proposal
was to search for a more complete theory (often called a hidden-variable
theory), whereas Bohr defended that no such theory could be satisfactorily
produced. However, things become more interesting when, following
Bell, we use angles that are different from simply measuring vertical
polarization (e.g. combinations of other directions). In the 1960's,
John Bell showed that (local) hidden-variable theories were incompatible
with the predictions of QM. Stating in the formalism we put forth,
Bell showed that if QM is correct, then some random-variables variables
representing the outcomes of spacelike-separated experiments are contextual.
About a decade later, Aspect, Grangier, and Gérard \cite{aspect_experimental_1981}
provided the first evidence that QM was correct, and recent (loophole-free)
experiments seem to corroborate their conclusions \cite{hensen_experimental_2015}. 

Bell's result comes out of the construction of a simple random variable
$\mathbf{S}$ defined as 
\begin{equation}
\mathbf{S}=\mathbf{AB}+\mathbf{A}'\mathbf{B}+\mathbf{A}\mathbf{B}'-\mathbf{A}'\mathbf{B}',\label{eq:S}
\end{equation}
where $\mathbf{A}$, $\mathbf{A}'$, $\mathbf{B}$, and $\mathbf{B}'$
are $\pm1$-valued random variables corresponding to outcomes of experiments
for Alice and Bob, with the prime denoting different spin-measurement
angles. It is straightforward to check that $\mathbf{S}$ can take
values $-2$ or $2$ (to verify this, one can make a table with all
$16$ possible values for $\mathbf{A}$, $\mathbf{A}'$, $\mathbf{B}$,
and $\mathbf{B}'$ and compute $\mathbf{S}$). Therefore, the expected
value of $\mathbf{S}$ must be between $-2$ and $2$, and we obtain
the Clauser-Horne-Shimony-Holt (CHSH) inequalities (permutations of
the minus sign gives you the other ones) \cite{clauser_proposed_1969}
\begin{equation}
-2\leq\left\langle \mathbf{AB}\right\rangle +\left\langle \mathbf{A}'\mathbf{B}\right\rangle +\left\langle \mathbf{A}\mathbf{B}'\right\rangle -\left\langle \mathbf{A}'\mathbf{B}'\right\rangle \leq2,\label{eq:CHSH}
\end{equation}
where here we introduce a shorter standard notation for expectation,
i.e. $E\left(\cdot\right)=\left\langle \cdot\right\rangle $. As in
the above example for the firefly box, if (\ref{eq:CHSH}) is violated,
there is no joint probability distribution, and the system of random
variables is contextual. 

The Bell-EPR setup differs significantly from the Kochen-Specker.
The random variables in Bell-EPR are necessarily consistently connected.
If they were not, it would be possible to used EPR-type correlated
systems to communicate superluminally: a choice of measurement direction
by Alice would instantly affect the mean value of Bob's measurements,
and she could use entangled particles to communicate with Bob. This
would be incompatible with the causal structure of special relativity,
and would require a complete rethinking of relativistic physics. Thus,
the absence of a joint probability distribution comes from the (non-trivial)
correlations imposed by the experimental outcomes (through the values
of the moments). But, more importantly, the Bell-EPR case provides
a situation where two parts of a system are correlated in ways that
cannot be explained by the existence of a common cause (hidden-variable)
because they are contextual. This is particularly disturbing to the
physicist because those two parts may be arbitrarily far away from
each other, and the events that are correlated may be spacelike separated.
A striking way to see how this is difficult to understand is if we
look at a firefly box-like construction for the Bell-EPR variables.
We will not attempt to do this here, as it would be lengthy, but we
refer the interested reader to an interesting paper by Blasiak \cite{blasiak_classical_2015}. 

Once again, a general solution may be obtained for the joint moments
in (\ref{eq:CHSH}), and we have the following (maybe negative, depending
on whether (\ref{eq:CHSH}) is violated or not) joint probability
distribution:
\begin{eqnarray*}
p\left(\omega_{aa'bb'}\right) & = & \frac{1}{4}\left(\left\langle \mathbf{AB}\right\rangle +\left\langle \mathbf{A}'\mathbf{B}\right\rangle +\left\langle \mathbf{A}\mathbf{B}'\right\rangle +\left\langle \mathbf{A}'\mathbf{B}'\right\rangle \right)+\alpha_{3}+\alpha_{4}-\alpha_{7}\\
\\
p\left(\omega_{aa'bb'}\right) & = & \frac{1}{4}\left(\left\langle \mathbf{AB}\right\rangle +\left\langle \mathbf{A}'\mathbf{B}\right\rangle \right)+\alpha_{3}+\alpha_{4}-\alpha_{7},\\
p\left(\omega_{aa'b\overline{b'}}\right) & = & \alpha_{7},\\
p\left(\omega_{aa'\overline{b}b'}\right) & = & -\frac{1}{4}\left(\left\langle \mathbf{AB}\right\rangle +\left\langle \mathbf{A}'\mathbf{B}\right\rangle -\left\langle \mathbf{A}\mathbf{B}'\right\rangle -\left\langle \mathbf{A}'\mathbf{B}'\right\rangle \right)+\alpha_{2}-\alpha_{3}+\alpha_{7},\\
p\left(\omega_{aa'\overline{b}\overline{b'}}\right) & = & -\frac{1}{4}\left(\left\langle \mathbf{A}\mathbf{B}'\right\rangle +\left\langle \mathbf{A}'\mathbf{B}'\right\rangle \right)+\alpha_{1}+\alpha_{3}-\alpha_{7},\\
p\left(\omega_{a\overline{a'}bb'}\right) & = & \frac{1}{4}\left(1-\left\langle \mathbf{A}'\mathbf{B}\right\rangle \right)-\alpha_{3}-\alpha_{4}-\alpha_{6},\\
p\left(\omega_{a\overline{a'}b\overline{b'}}\right) & = & \alpha_{6},\\
p\left(\omega_{a\overline{a'}\overline{b}b'}\right) & = & \frac{1}{4}\left(\left\langle \mathbf{A}'\mathbf{B}\right\rangle -\left\langle \mathbf{A}'\mathbf{B}'\right\rangle \right)-\alpha_{2}+\alpha_{3}+\alpha_{6},\\
p\left(\omega_{a\overline{a'}\overline{b}\overline{b'}}\right) & = & \frac{1}{4}\left(1+\left\langle \mathbf{A}'\mathbf{B}'\right\rangle \right)-\alpha_{1}-\alpha_{3}-\alpha_{6},\\
p\left(\omega_{\overline{a}a'bb'}\right) & = & -\frac{1}{4}\left(\left\langle \mathbf{AB}\right\rangle +\left\langle \mathbf{A}\mathbf{B}'\right\rangle \right)+\alpha_{1}-\alpha_{4}+\alpha_{5},\\
p\left(\omega_{\overline{a}a'b\overline{b'}}\right) & = & \frac{1}{4}\left(1+\left\langle \mathbf{A}\mathbf{B}'\right\rangle \right)-\alpha_{1}-\alpha_{3}-\alpha_{5},\\
p\left(\omega_{\overline{a}a'\overline{b}b'}\right) & = & \frac{1}{4}\left(1+\left\langle \mathbf{AB}\right\rangle \right)-\alpha_{1}-\alpha_{2}-\alpha_{5},\\
p\left(\omega_{\overline{a}a'\overline{b}\overline{b'}}\right) & = & \alpha_{5},\\
p\left(\omega_{\overline{a}\overline{a'}bb'}\right) & = & \alpha_{4},\\
p\left(\omega_{\overline{a}\overline{a'}b\overline{b'}}\right) & = & \alpha_{3},\\
p\left(\omega_{\overline{a}\overline{a'}\overline{b}b'}\right) & = & \alpha_{2},\\
p\left(\omega_{\overline{a}\overline{a'}\overline{b}\overline{b'}}\right) & = & \alpha_{1},
\end{eqnarray*}
where $\alpha_{i}$ are free parameters. Once again, the ranges of
$\alpha_{i}$ depend on the values of the moments, but one point is
relevant here. On the $\mathbf{A}$, $\mathbf{B}$, $\mathbf{C}$
example, we had only one free parameter, while here we have seven.
The reason is that in the Bell-EPR setup, only the four individual
expectations and four moments are given, and together with the requirement
that $\sum p\left(\omega_{i}\right)=1$ this amounts to 9 equations
for sixteen elementary events, thus it is a more underdetermined case.

\subsection{Quantum contextuality in psychology}

We did not try to give an exhaustive list of all contextual systems
in QM, but mainly those which provide further conceptual understanding
of the difficulties faced by physicists trying to understand quantum
theory. We present those examples to provide a background for the
discussion of contextuality in psychology, a theme that is at the
core of current attempts to use the mathematics of QM to model cognition.
Once again, we will not try to give an exhaustive account of all different
contextual cases, and the interested reader is referred to Busemeyer
and Bruza's book \cite{busemeyer_quantum_2012}. Here we briefly examine
a few cases that exemplify quantum-like contexts in psychology: violations
of the sure-thing-principle \cite{aerts_quantum_2009,khrennikov_quantum_2009}
in decision making and order effects \cite{wang_context_2014}. 

Savage's Sure-Thing-Principle was stated the following way \cite[pg. 21]{savage_foundations_1972}: 
\begin{quotation}
``A businessman contemplates buying a certain piece of property.
He considers the outcome of the next presidential election relevant
to the attractiveness of the purchase. So, to clarify the matter for
himself, he asks whether he should buy if he knew that the Republican
candidate were going to win, and decides that he would do so. Similarly,
he considers whether he would buy if he knew that the Democratic candidate
were going to win, and again finds that he would do so. Seeing that
he would buy in either event, he decides that he should buy, even
though he does not know which event obtains, or will obtain, as we
would ordinarily say. It is all too seldom that a decision can be
arrived at on the basis of the principle used by this businessman,
but, except possibly for the assumption of simple ordering, I know
of no other extralogical principle governing decisions that finds
such ready acceptance.''
\end{quotation}
For example, imagine you have $\mathbf{B}$ and $\mathbf{P}$ as a
$\pm1$-valued random variables corresponding to ``not buy'' ($\mathbf{B}=-1$)
or ``buy'' ($\mathbf{B}=+1$), and ``Republican president'' ($\mathbf{P}=-1$)
or ``Democrat president'' ($\mathbf{P}=+1$). The STP corresponds
to the probabilistic statement that 
\begin{eqnarray*}
P\left(\mathbf{B}=1\right) & = & P\left(\mathbf{B}=1|\mathbf{P}=1\right)P\left(\mathbf{P}=1\right)+P\left(\mathbf{B}=1|\mathbf{P}=-1\right)P\left(\mathbf{P}=-1\right)\\
 & \geq & P\left(\mathbf{B}=-1|\mathbf{P}=1\right)P\left(\mathbf{P}=1\right)+P\left(\mathbf{B}=-1|\mathbf{P}=-1\right)P\left(\mathbf{P}=-1\right)\\
 & = & P\left(\mathbf{B}=-1\right),
\end{eqnarray*}
if 
\[
P\left(\mathbf{B}=1|\mathbf{P}=1\right)\geq P\left(\mathbf{B}=-1|\mathbf{P}=1\right)
\]
and 
\[
P\left(\mathbf{B}=1|\mathbf{P}=1\right)\geq P\left(\mathbf{B}=-1|\mathbf{P}=1\right).
\]

Tversky and Shafir showed that human decision makers often do not
follow the STP \cite{shafir_thinking_1992,tversky_disjunction_1992}.
Since STP follows in a straightforward way from the axioms of probability
theory, violations of STP by human decision makers imply they do not
follow those axioms themselves, but perhaps some type of generalized
probability theory. Such generalized probability theory, as some have
proposed, is the one given by probabilities defined over an orthomodular
lattice resulting from measures over a Hilbert space, i.e., quantum
probabilities \cite{blutner_quantum_2015}. 

For example, the STP can be given by a quantum description of the
MZI paradigm \cite{de_barros_examples_2015}. In the MZI, where which
path information causes a collapse of the wave function, therefore
changing the probability distributions of the outcomes of the experiment.
So, if we use the analogy that in the MZI the responses ``buy''
or ``not buy'' correspond to detectors at the end of the interferometer,
and the which-path information corresponding to ``Republican president''
or ``Democrat president,'' the collapse of the wave function would
change the distributions of ``buy'' or ``not buy'' depending on
the context of knowing which is president, similar to the Tversky
and Shafir's experiments. Therefore, violations of STP show a clear
case of contextuality. However, it is obviously trivial contextuality,
since the ``measurement'' of which-path creates a direct change
in the expectation values of the ``buy''/``not buy'' random variable. 

We now turn to order effects. Order effects are well-known in quantum
systems, where successive measurements of incompatible quantities
(e.g. spin in two orthogonal directions) give different results depending
on the order. Recently, in a model similar to the quantum model for
the MZI, Wang et al. \cite{wang_context_2014}, showed that not only
can quantum models correctly reproduce the observed order effect of
outcomes of many different experiments, but they can also predict
a non-trivial relation for the order effect: the QQ equality. This
equality, which holds exactly for the quantum formalism, seems to
also hold with good fit for most order effect experiments investigated
by Wang et al., a surprising finding, since it seems the QQ equality
cannot be derived in any straightforward way from other approaches.
However, as in the STP example, the random variables are inconsistently
connected.

\section{Final remarks\label{sec:Final-remarks}}

In this paper we described negative probabilities, and showed how
they can be used to describe some contextual systems. We tried to
show in the examples some of the cases where negative probabilities
work well, but also those where no clear approach with negative probabilities
exist (i.e. for inconsistently connected systems). Our goal was to
provide a different approach to contextual systems than the formalism
of Quantum Mechanics, one that may perhaps be useful in quantum cognition.
The advantage of NP is that it can model not only those situations
where QM is applied, but it is also more general. 

As an example, let us think about the three-sided firefly box. In
QM, if we have three observables that can be observed simultaneously
in pairs, it follows that they can also be observed all together.
This is a characteristic of the Hilbert space formalism, and can be
easily demonstrated (see \cite{de_barros_joint_2012,de_barros_beyond_2015}).
However, it is also possible to show that, under certain reasonable
assumptions, one should expect a neural stimulus-response model to
be able to reproduce the types of correlations that we find in the
three random variable case, where no joint probability distribution
exists \cite{de_barros_quantum-like_2012,de_barros_quantum_2015}.
Thus, we are left with the possibility of a plausible contextual system
that is forbidden by the quantum formalism and that can easily be
described by NP, as we saw in Section \ref{sec:Some-examples}. Additionally,
as mentioned earlier, there are many surprising theorems in QM that
seem to have no counterparts in psychology or social sciences, and
a more general contextual theory of probabilities might be advantageous.

We should point out that despite all the discussions about contextuality
in social systems, recently Dzhafarov, Zhang, and Kujala analyzed
many psychology experiments, and found no evidence of non-trivial
contextuality \cite{dzhafarov_is_2015}. This means that more subtle
examples, such as the firefly box or systems equivalent to the Bell-EPR
where contextuality comes from the correlations and not from inconsistently
connected random variables, were not found. Their analysis was made
using the apparatus of Contextuality-by-Default, an approach that
is more general than the NP. This does not mean that NP are not necessarily
useful in the social sciences, but it seems that up to now attempts
to find non-trivially contextual systems have failed. 

As we saw in the examples, as well as in the discussions that followed
Proposition \ref{prop:m=00003D1impliesjoint}, the minimum value of
the L1 norm, $M^{*}$, can be interpreted as a measure of contextuality
\cite{de_barros_measuring_2014}. This is connected to standard views
in QM, where the values of $\left\langle \mathbf{S}\right\rangle $
(equation (\ref{eq:S})) are taken as a measure of departure from
locality for Bell-EPR systems, with higher values of $\left\langle \mathbf{S}\right\rangle $
corresponding to more non-local systems (therefore more contextual).
This is also true for the three random variable system $\mathbf{A}$,
$\mathbf{B}$, and $\mathbf{C}$, where $M^{*}$ is associated to
the expectation of $\mathbf{AB}+\mathbf{BC}+\mathbf{AC}$ present
in equation (\ref{eq:suppes-zanotti}). It would be interesting to
see how $M^{*}$ compares to other measures of contextuality, namely
the one given by Contextuality-by-Default, for more complex systems,
and whether interesting classifications can arise from different measures
of contextuality.

\paragraph{Acknowledgments. }

The authors are indebted to Ehtibar Dzhafarov for many discussions
on contextuality and to comments on this manuscript. JAB and GO would
like to acknowledge the influence and support of Patrick Suppes on
developing the theory of NP, which he became fond of in the last years
before his death. 

\bibliographystyle{plain}
\bibliography{Quantum}

\end{document}